\documentclass[12pt, a4paper]{amsart}

\newcommand{\Ignore}[1]{\ignorespaces}

\usepackage{amsmath, amssymb}
\usepackage{graphicx}
\usepackage{multirow}
\usepackage{hyperref}

\newtheorem{theorem}{Theorem}[section]
\newtheorem{lemma}[theorem]{Lemma}
\newtheorem{observation}[theorem]{Observation}
\newtheorem{corollary}[theorem]{Corollary}

\newtheorem{conjecture}[theorem]{Conjecture}

\newcommand{\Reals}{\mathbb{R}}
\newcommand{\DM}{\ensuremath{\mathtt{DM}}}
\newcommand{\CritG}{\ensuremath{\mathrm{C}_{\mathit{g}}}}
\newcommand{\CritB}{\ensuremath{\mathrm{C}_{\mathit{b}}}}
\newcommand{\CritP}{\ensuremath{\mathrm{C}_{\mathit{p}}}}
\newcommand{\CritGw}{\ensuremath{\mathrm{C}_{\mathit{g}}^{\mathit{w}}}}
\newcommand{\CritBw}{\ensuremath{\mathrm{C}_{\mathit{b}}^{\mathit{w}}}}
\newcommand{\Alg}{\ensuremath{\mathtt{Alg}}}

\begin{document}
	
\title{Stable Secretaries}
\thanks{This paper will be presented at the \emph{18th ACM conference on Economics and Computation (EC 2017)}.}

\author[Babichenko]{Yakov Babichenko}
\address{Technion, Israel}
\email{yakovbab@tx.technion.ac.il}
\thanks{Y. Babichenko was supported by the Israel Science Foundation grant number 2021296.}

\author[Emek]{Yuval Emek}
\address{Technion, Israel}
\email{yemek@technion.ac.il}

\author[Feldman]{Michal Feldman}
\address{Tel Aviv University, Israel}
\email{mfeldman@tau.ac.il}
\thanks{The work of M. Feldman was partially supported by the European Research Council under the European Union's Seventh Framework Programme (FP7/2007-2013) / ERC grant agreement number 337122.}

\author[Pratt-Shamir]{Boaz Patt-Shamir}
\address{Tel Aviv University, Israel}
\email{boaz@tau.ac.il}

\author[Peretz]{Ron Peretz}
\address{Bar Ilan University, Israel}
\email{ron.peretz@biu.ac.il}

\author[Smorodinsky]{Rann Smorodinsky}
\address{Technion, Israel}
\email{rann@ie.technion.ac.il}
\thanks{R. Smorodinsky was supported by Technion VPR grants, the Bernard M. Gordon Center for Systems Engineering at the Technion, and the TASP Center at the Technion.}

\begin{abstract}
We define and study a new variant of the secretary problem.
Whereas in the classic setting multiple secretaries compete for a single
position, we study the case where the secretaries arrive one at a time and are
assigned, in an on-line fashion, to one of \emph{multiple} positions.
Secretaries are ranked according to talent, as in the original formulation,
and in addition positions are ranked according to attractiveness.
To evaluate an online matching mechanism, we use the notion of \emph{blocking
pairs} from stable matching theory:
our goal is to maximize the number of positions (or secretaries) that do not
take part in a blocking pair.
This is compared with a \emph{stable matching} in which no blocking pair
exists.
We consider the case where secretaries arrive randomly, as well as that of an
adversarial arrival order, and provide corresponding upper and lower bounds.
\end{abstract}

\maketitle

\section{Introduction}
The celebrated \emph{secretary problem}, which first appeared in print in
Martin Gardner's 1960 \emph{Scientific American} column \cite{G60} (but
apparently originated much earlier, see \cite{F89}), considers a simple online
problem where multiple applicants interview sequentially for an open position
(say a secretary).
The interview of an applicant allows the employer to assess the relative
quality of this applicant with respect to all those interviewed so far.
An irreversible decision whether to hire or reject an applicant must
be made as soon as the applicant's interview is over.
In particular, the decision is taken without knowing anything about the
quality of future applicants.
The problem is then to find a scheme that will maximize the probability of
choosing the best applicant.

The problem gained considerable popularity and many variants have been
introduced and studied since.
Some of the classical extensions have been
(1) to move from an \emph{ordinal} setting, where applicants are ranked one
compared to the other (a.k.a.\ the \emph{comparison-based} model), to a
\emph{cardinal} setting, where applicants are identified with an absolute
score;
(2) to generalize the original \emph{uniform} arrival order to other
distributions as well as to an adversarial arrival order;
(3) to choose an applicant from the top tier instead of necessarily the best
one;
and
(4) to choose more than one secretary.
For more information, the reader is referred to Ferguson \cite{F89} for
an early survey  of the history of the secretary problem and to
Dinitiz \cite{D13} who gives a survey of relatively recent results with an
emphasis on applications to auction theory.

For the most part, the models and extensions considered throughout the years
maintain the invariant that the hiring challenge involves a single
position or, in some generalizations, several identical positions.%
\footnote{The unique exception we are familiar with is the decentralized model
considered in \cite{ChenHKLM2015}.}
However, many practical applications call for the need to perform many-to-many
matchings, where candidates arrive sequentially and are matched to a fixed
pool of \emph{non-identical} positions.
For example, in the labor market it is often the case that multiple applicants
arrive sequentially and are assigned to various existing job openings within a
firm.
Likewise,
in some dating services, men arrive sequentially and are matched to a fixed
pool of women,
in transportation services (e.g., Uber), passengers arrive sequentially (say,
getting out of an airport terminal) and are matched to a pool of drivers,
and
in a reviewing process for journals, the manuscripts
arrive sequentially and are assigned to the various members of the editorial
board.

In this paper, we introduce a new version of the secretary problem, where the
primary novelty is to consider multiple non-identical positions.
Applicants are matched to one of the positions (or left jobless) by some
central authority, e.g., a human resources department.
As in the original problem, applicants are interviewed sequentially and their
relative rank among applicants interviewed thus far is determined (namely, we
consider an ordinal setting).
At the end of the interview, the applicant is assigned to (at most) one
position in an irreversible manner.

The challenge is to match applicants to positions in an optimal way, but what
is the optimality criterion when non-identical positions are involved?
We augment the standard model by adding a preference order over the various
positions.
In other words, positions are not equally attractive (e.g., the related
salaries are different), however any position is preferred to unemployment.
To evaluate an online matching mechanism, we use the notion of \emph{blocking
pairs} from stable matching theory (\cite{GS62,GI89}).

To explain this notion, consider an arbitrary matching.
A pair, made of a position and an applicant, is said to form a \emph{blocking
pair} if they both prefer each other over those they are matched with.
We would like to maximize the number of positions (or applicants) that do
\emph{not} take part in a blocking pair.
This is compared against an optimal off-line \emph{stable matching} that
admits no blocking pairs.%
\footnote{The original formulation of the secretary problem is equivalent to
finding the unique stable matching, assuming that all secretaries prefer
employment to unemployment.}
An alternative objective is to maximize the probability of producing a
stable matching;
this, however, is virtually impossible as one can easily show that
when the applicants' arrival order is chosen uniformly at random, the
probability to generate a stable matching decreases exponentially with the
number of positions/applicants.

We consider three variants of the problem inspired by potential underlying
business models of the central matching agency:
in some cases, it is the employer that pays (this is indeed typical to the
labor market);
in other cases, it is the sequentially arriving applicants that pay (the
common practice in some professional companionship services);%
\footnote{Such services for matching professional companions for the elderly
are popular in some countries.}
and finally, in other cases the matching agency is interested in satisfying
both sides of the market.
We study scenarios with either random or adversarial applicant arrival order
and provide corresponding upper and lower bounds.

\subsection{The Model.}
We now turn to a formal exposition of the abstract model investigated in this
paper.
Throughout, we consider finite totally ordered \emph{girl} set $G$
(corresponding to the positions or the employers behind them) and \emph{boy}
set $B$ (corresponding to the applicants) with $\succ$ denoting the order
relation referred to hereafter as a \emph{preference}.
While the girls and their total order are known to the decision maker
(denoted by \DM{}) in advance, the boys arrive in an \emph{online fashion} so
that boy
$\pi(t) \in B$
arrives at time $t$ for
$t = 1, \dots, |B|$,
where $\pi$ is an (initially hidden) permutation over $B$.
Unless stated otherwise, it is assumed that the number $|B|$ of boys is known
to \DM{} in advance.

Upon arrival of boy
$b = \pi(t)$,
its relative rank among boys
$\pi(1), \dots, \pi(t)$
is reported to \DM{}.
In response, \DM{} matches $b$ to some girl in $G$ that was not matched
beforehand or leaves $b$ unmatched;
this (un)matching operation is irrevocable.
We treat an unmatched individual as being matched to the designated symbol
$\bot$ and extend the definition of the preference $\succ$ so that
$x \succ \bot$
for every individual
$x \in G \cup B$.

Consider the situation after all boys have arrived.
Girl $g$ and boy $b$ form a \emph{blocking pair} if
$g \succ g(b)$
and
$b \succ b(g)$,
where $g(b)$ (resp., $b(g)$) denotes the girl (resp., boy) matched to boy $b$
(resp., girl $g$) or $\bot$ if $b$ (resp., $g$) was left unmatched.
Girl $g$ (resp., boy $b$) is said to be \emph{satisfied} if she (resp., he)
is matched and does not participate in a blocking pair.
A $(g, b)$ matched pair is said to be \emph{satisfied} if both $g$ and $b$ are
satisfied.
The objective of \DM{} is to maximize one of the following three criteria: \\
\underline{\CritG{}:} the number of satisfied girls; \\
\underline{\CritB{}:} the number of satisfied boys; or \\
\underline{\CritP{}:} the number of satisfied matched pairs. \\
Since a stable matching induces
$n = \min\{ |G|, |B| \}$
stable pairs, \DM{} aims for algorithms that guarantee to satisfy (in
expectation)
$\rho n$
girls (\CritG{}), boys (\CritB{}), or matched pairs (\CritP{}) for as large as possible
approximation ratio $\rho$, typically expressed as
$\rho = \rho(n)$.

The aforementioned setting can be generalized by augmenting the girls (resp.,
boys) with a \emph{weight} function
$w : G \rightarrow \Reals_{> 0}$
(resp.,
$w : B \rightarrow \Reals_{> 0}$).
In that case, the objective of \DM{} is to maximize one of the following two
criteria: \\
\underline{\CritGw{}:} the total weight of satisfied girls; or \\
\underline{\CritBw{}:} the total weight of satisfied boys. \\
(The weighted version of satisfying matched pairs is not treated in this
paper.)
Taking
$H_{G} \subseteq G$
(resp.,
$H_{B} \subseteq B$)
to be the subset consisting of the
$n = \min \{ |G|, |B| \}$
heaviest (in terms of $w$, breaking ties arbitrarily) girls (resp.,
boys), we observe that the total weight of satisfied girls (resp.,
boys) in an optimal (stable) matching is $w(H_{G})$ (resp., $w(H_{B})$).%
\footnote{We follow the convention that the weight of a set is the total
weight of the set's elements.}
Therefore, in the weighted setting, \DM{} aims for algorithms that guarantee
to satisfy girls (resp., boys) whose total weight (in expectation) is
$\rho w(H_{G})$
(resp.,
$\rho w(H_{B})$)
for as large as possible approximation ratio $\rho$, typically expressed as
$\rho = \rho(n)$.

It will be convenient to assume that the boy arrival permutation $\pi$ is
chosen according to some probability distribution $\Pi$.
When not stated otherwise, $\Pi$ is assumed to be uniform, but
we also consider the case where $\Pi$ is designed by an oblivious adversary
(this is restricted to Section~\ref{section:adversarial-arrival-order}).
The guarantee of \DM{}'s algorithmic strategy is taken in expectation over the
distribution $\Pi$ and possibly also over the random coin tosses of \DM{} (if
it is randomized).
To avoid confusion, we emphasize that the order relation $\succ$
and weights $w$ (in the weighted version) are determined \emph{before} the
random choice of $\pi$ is performed (say, by a designated nature player).

\subsection{Our Contribution.}
Our contribution is both conceptual and technical.
Conceptually, we consider the problem of a central authority that assigns
applicants to one of many non-identical positions.
Allowing a variety of positions introduces the challenge of identifying the
criterion by which one should measure the quality of a match.
We propose some formal criteria, inspired by the concept of stable marriage,
to measure the quality of an online assignment.

Beyond the conceptual contribution suggested above, we provide upper and lower
bounds on the performance of online assignment algorithms (refer to
Table~\ref{table:results} for a summary).
Our main results concern the unweighted case:
when the arrival order is random (distributed uniformly), one can satisfy
$\Omega (n)$ positions/applicants (corresponding to girls/boys,
Theorem~\ref{theorem:positive-Cg-Cb}), however, no (randomized) algorithm can
guarantee more than $O(\sqrt n)$ satisfied matched pairs
(Theorem~\ref{theorem:negative-Cp});
on the other hand, if the arrival order is adversarial, then any (randomized)
algorithm can satisfy at most $O(\sqrt n)$ positions/applicants
(Theorem~\ref{theorem:negative-adversarial-Cg-Cb}) and at most $1$ matched pair
(Theorem~\ref{theorem:negative-adversarial-Cp}).
We further consider the case of weighted candidates and positions.
Here, we show that the total weight of satisfied positions in an optimal
(stable) matching can be approximated within an
$\Omega (1 / \log n)$
ratio (Theorem~\ref{theorem:positive-wCg}) and the total weight of satisfied
applicants in an optimal (stable) matching can be approximated within an
$\Omega (1)$
ratio (Theorem~\ref{theorem:positive-wCb}).

\begin{table}
\begin{center}
\begin{tabular}{|c|c|c|c|c|c|c|}
\hline
\multicolumn{2}{|c|}{}
&
\multicolumn{5}{|c|}{optimization criterion}
\\ \cline{3-7}
\multicolumn{2}{|c|}{}
&
\CritG{} & \CritB{} & \CritP{} & \CritGw{} & \CritBw{}
\\ \hline
\multirow{4}{*}{arrival}
&
\multirow{2}{*}{uniform random}
&
$\Omega (1)$
&
$\Omega (1)$
&
$O (1 / \sqrt{n})$
&
$\Omega (1 / \log n)$
&
$\Omega (1)$
\\
& &
[\ref{theorem:positive-Cg-Cb}]
&
[\ref{theorem:positive-Cg-Cb}]
&
[\ref{theorem:negative-Cp}]
&
[\ref{theorem:positive-wCg}]
&
[\ref{theorem:positive-wCb}]
\\ \cline{2-7}
&
\multirow{2}{*}{adversarial}
&
$O (1 / \sqrt{n})$
&
$O (1 / \sqrt{n})$
&
1
& & \\
& &
[\ref{theorem:negative-adversarial-Cg-Cb}]
&
[\ref{theorem:negative-adversarial-Cg-Cb}]
&
[\ref{theorem:negative-adversarial-Cp}]
& &
\\ \hline
\end{tabular}
\end{center}
\caption{\label{table:results}%
Our main technical results.
Each cell specifies a bound on the achievable approximation ratio
$\rho = \rho(n)$
for the given optimization criterion (columns) and arrival order distribution
(rows).
The corresponding theorem numbers are specified in brackets.
}
\end{table}

\subsection{Related Work.}
In contrast to our optimization criteria that aim at maximizing the number of
satisfied positions (or applicants or matched pairs), i.e., positions that do
not participate in a blocking pair, some previous papers
\cite{KMV94,ABM05,EH08,OR15} follow the complement approach of minimizing the
number of blocking pairs (which can be quadratic).
Notice that these two criteria are very different as a single agent may
contribute to multiple blocking pairs.
We believe that our approach is more natural when looking at our problems from
the perspective of generalizing the classic secretary problem:
instead of aiming at satisfying a single position (the classic secretary
problem), we now try to satisfy as many positions as possible out of $n$
non-identical positions.

Recall that most of our technical focus is dedicated to the balanced scenario
where
$|G| = |B|$.
Interestingly, in the static model the balanced scenario is a knife-edge case
for some phenomena such as multiplicity of stable outcomes (e.g.,
\cite{AshlagiKL2016}).
This is in contrast to our online setting, where unbalanced scenarios can be
reduced to the balanced one, possibly with a constant loss in the guaranteed
approximation ratio.

Optimization criterion \CritP{} can be compared to the one considered by
~\cite{HMV16}.
In \cite{HMV16}, a subset of the participants can be ignored (i.e., not
matched) and the matching is required to be stable with respect to the matched
participants.
The objective is then to minimize the number of omitted participants.

The economic literature on \emph{dynamic matching} often focuses on the tension
between market thickness and participants' waiting time.%
\footnote{Not to be confused with the algorithmic literature on the maximum
matching problem in dynamic graphs.}
The idea is that participants join the market and can be matched thereafter
at any given point of time.
The longer one waits with the matching, the thicker the market becomes and so
it may be possible to find a better match.
On the other hand, participants may lose utility due to the waiting time
(e.g., health deterioration in the context of the kidney matching market).
Some examples that study this tension are \cite{BLY15} and \cite{AshlagiBJM16} 
(see also \cite{EmekKW2016} for a related online matching formulation).
In contrast with our work where the number of agents is finite and hence a
hindsight benchmark for comparing the outcome of an online mechanism  is
natural, these models consider an infinite stream of agents whose preferences
are stochastic, generated by a stationary source.
The objective function in this case is not to minimize some criterion in
hindsight, but rather to maximize the total expected utility, taking into
account both the utility from each match and the agents' waiting times.

The online nature of the maximization problems studied in the current paper is
inspired by the online bipartite matching model of Karp, Vazirani, and Vazirani
\cite{KarpVV1990}, where the nodes in one side of a bipartite graph are known
from the beginning and the nodes in the other side arrive in an online
fashion together with their incident edges.
This model became very popular with quite a few papers aiming at maximizing
the size of the matching
\cite{KarpVV1990, BirnbaumM2008, GoelM2008, DevanurJK2013, Miyazaki2014,
NaorW2015}
or the weight of the matched nodes on the static side
\cite{AggarwalGKM2011, DevanurJK2013, NaorW2015}.
\cite{AggarwalGKM2011} also show that in the general case, no
online algorithm can guarantee a non-trivial competitive ratio on the weight
of the edges included in the output matching (the weighted nodes setting is a
special case of weighted edges, where all edges incident on the same static
node admit the same weight).
In contrast, \cite{KesselheimRTV13} prove that under a random
arrival order, this problem can be approximated within a constant factor.
Notice that the graph topology and the edge weights (if the edges are
weighted) implicitly induce a set of cardinal preferences, where a heavier
edge is preferable to a lighter one (the weighted case) and any edge is
preferable to no edge at all.
However, the preferences that can be defined this way are inherently
symmetric and as such, form a strict subset of the (ordinal) preferences
considered in the current paper.

A different line of work studied variants of the secretary problem where the
algorithm designer should select any subset of the arriving candidates subject
to some combinatorial constraints.
This line of research has received significant attention recently, in part due
to applications to auction theory and mechanism design.
The most famous
variant is the matroid secretary problem, where the chosen subset forms an
independent set in a matroid \cite{BIK07,L14,FSV15}, but recent work
considered more general combinatorial constraints as well \cite{R16}.
Secretary settings with non-uniform random arrival orders have been
investigated in \cite{KKN15}.

\section{General Transformations}
\label{section:general-transformations}
We begin with ``black-box'' lemmas that help us develop a better understanding
of the different optimization criteria.
To that end, we say that the matching algorithm is \emph{conservative} if it
is guaranteed to output a (size-wise) maximum matching.
Alternatively, a conservative algorithm may decide to leave a boy unmatched
only if the number of pending boys is at least as large as the number of yet
unmatched girls.

\begin{lemma} \label{lemma:conservative-Cb-Cp}
Optimization criteria \CritBw{} and \CritP{} admit optimal conservative
approximation algorithms.
This holds for every arrival order distribution.
\end{lemma}
\begin{proof}
We establish the assertion for optimization criterion \CritBw{};
the proof for optimization criterion \CritP{} is based on the same line of
arguments.
Let
$b_{1}, \dots, b_{n}$
be the boys indexed in order of arrival.
For
$1 \leq i \leq n$,
let
$B_{i} = \{ b_{i}, b_{i + 1}, \dots, b_{n} \}$
and let $G_{i}$ be the set of girls that are unmatched upon arrival of boy
$b_{i}$.
Matching boy $b_{i}$ to girl
$g \in G$
is said to be a \emph{weak matching} action if
$|G_{i}| > |B_{i}|$
and $g$ is among the
$|G_{i}| - |B_{i}|$
weakest available girls.\footnote{%
Throughout, the terms weak and strong refer to the preference order $\succ$ in
the natural manner.
}

We first argue that optimization criterion \CritBw{} admits an optimal
algorithm that never performs weak matching actions.
To that end, consider some algorithm \Alg{} that performs a weak matching
action by matching boy $b_{i}$ to girl $g$ and let
$S \subseteq G_{i}$
be the set of the $|B_{i}|$ strongest unmatched girls upon arrival of $b_{i}$.
By definition, at least one of the girls in
$g_{s} \in S$
is unmatched in the final outcome of \Alg{}, hence all boys matched to girls
$g' \prec g_{s}$,
including $b_{i}$,
are unsatisfied.
Therefore, the algorithm that mimics \Alg{} at all times other than $i$ and
leaves $b_{i}$ unmatched at time $i$ satisfies the same set of boys that
\Alg{} does.
By repeating this argument over all such times $i$, we come up with an
algorithm that does not perform any weak matching actions and satisfies the
same set of boys as \Alg{}.

So, let \Alg{} be an optimal \CritBw{}-algorithm that never performs weak
matching actions.
Consider some time $i$ such that
$|B_{i}| \leq |G_{i}|$
and let
$g_{1} \succ g_{2} \succ \cdots \succ g_{n - i + 1}$
be the $|B_{i}|$ strongest unmatched girls upon arrival of boy $b_{i}$.
Suppose that \Alg{} leaves $b_{i}$ unmatched.
Since \Alg{} never performs weak matching actions, it follows that girl
$g_{n - i + 1}$
will remain unmatched under \Alg{}.
Therefore, an algorithm that mimics \Alg{} at all times other than $i$ and
matches $b_{i}$ to
$g_{n - i + 1}$
is guaranteed to satisfy all the boys that \Alg{} satisfies.
By repeating this argument over all such times $i$, we obtain an optimal
conservative algorithm.
\end{proof}

\begin{conjecture}
Optimization criterion \CritG{} does not admit an optimal conservative
approximation algorithm under a random arrival order for any sufficiently
large
$n = |B| = |G|$.
\end{conjecture}

Next, we turn our attention to \emph{balanced} instances, where
$|G| = |B|$
and prove that these instances are as hard (up to a constant factor) as the
general case for optimization criteria \CritG{} and \CritBw{}.
Moreover, in balanced instances, optimization criteria \CritG{} and \CritB{}
are, in fact, equivalent for conservative algorithms despite the inherent
asymmetry between girls and boys in our setting.

\begin{lemma} \label{lemma:reduce-to-balanced-girls}
There exist universal constants
$\alpha, \beta > 0$
such that
an algorithm that guarantees to approximate optimization criteria \CritG{}
within ratio
$\rho = \rho(n)$
in an instance with
$|G| = |B| = n$
implies an algorithm that guarantees to approximate optimization criteria
\CritG{} within ratio
$\alpha \rho(\beta n)$
in an instance with
$\min \{ |G|, |B| \} = n$.
\end{lemma}

\begin{lemma} \label{lemma:reduce-to-balanced-boys}
There exist universal constants
$\alpha, \beta > 0$
such that
an algorithm that guarantees to approximate optimization criteria
\CritBw{} within ratio
$\rho = \rho(n)$
in an instance with
$|G| = |B| = n$
implies an algorithm that guarantees to approximate optimization
criteria \CritBw{} within ratio
$\alpha \rho(\beta n)$
in an instance with
$\min \{ |G|, |B| \} = n$.
\end{lemma}

The proofs of Lemmas \ref{lemma:reduce-to-balanced-girls} and
\ref{lemma:reduce-to-balanced-boys} rely on the following observation (whose
proof is deferred to Appendix~\ref{appendix:missing-proofs})

\begin{observation} \label{observation:good-event}
There exists a constant
$c > 0$
such that for every sufficiently large integer $k$ and for every integer
$c \leq \ell \leq k / c$,
if $\pi$ be a (uniform) random permutation over $[k]$ and
$R = \min\{ \pi(1), \dots, \pi(\lceil k / \ell \rceil) \}$,
then
\[
\Pr(\ell / 5 < R \leq \ell) > 1 / 13 \, .
\]
\end{observation}

\begin{proof}[Proof of Lemma~\ref{lemma:reduce-to-balanced-girls}]
If
$|G| = m > n = |B|$,
then we simply ignore any subset of
$m - n$
girls (leaving them unmatched) and run the algorithm promised by the
assumption on the remaining $n$ girls and all boys in $B$.

The more interesting case is when
$|G| = n < m = |B|$.
Let $c$ be the constant from
Observation~\ref{observation:good-event}.
If
$n < c$,
then it suffices to satisfy a single girl which can be fulfilled by applying
the classic secretary algorithm to the instance consisting of an arbitrary
girl
$g \in G$
and all boys in $B$,
thus satisfying $g$ with probability that converges to
$1 / e$
as
$m \rightarrow \infty$;
assume hereafter that
$n \geq c$.
We can further assume that
$n \leq m / c$
as otherwise, we simply ignore an arbitrary subset of
$n - m / c$
girls (leaving them unmatched).

Refer to the first
$\lceil m / n \rceil$
boys as the \emph{filter} boys and leave them unmatched.
Let $b_{f}$ be the most preferred filter boy and define the random set
$X = \{ b \in B \mid b \succ b_{f} \}$.
Observation~\ref{observation:good-event} ensures that the event
$n / 5 \leq |X| < n$
occurs with probability at least
$1 / 13$;
condition hereafter on this event.

Refer to the
$n / 5$
least preferred girls in $G$ as the \emph{target} girls.
We run the algorithm promised by the assumption on the target girls and the
first
$n / 5$
boys to arrive from $X$, matching any remaining boy from $X$ to an arbitrary
non-target girl and ignoring all boys not in $X$ (leaving them unmatched).
The assumption ensures that a
$\rho(n / 5)$
fraction of the target girls will be satisfied.
\end{proof}

\begin{proof}[Proof of Lemma~\ref{lemma:reduce-to-balanced-boys}]
If
$|G| = m > n = |B|$,
then one can simply ignore the
$m - n$
least preferred girls (leaving them unmatched) and run the algorithm promised
by the assumption on the remaining $n$ girls in $G$ and all boys in $B$.

The more interesting case is when
$|G| = n < m = |B|$.
Let $c$ be the constant promised by
Observation~\ref{observation:good-event} and fix
$\ell = 5 n$.
If
$\ell < c$,
then it suffices to satisfy the heaviest boy which can be trivially fulfilled
by matching him to the most preferred girl;
assume hereafter that
$\ell \geq c$.
We can further assume that
$\ell \leq m / c$
as otherwise, we simply ignore the last arriving
$m - \ell < m (1 - 1 / c)$
boys (leaving them unmatched), thus losing an expected weight of
$w(B) / c$;
employing Markov's inequality, we can condition on the lost weight to be
sufficiently close to it.

Our proof requires an assumption on the weights as well:
We assume that
$w(b) \neq w(b')$
for every two boys
$b, b' \in B$;
this is without loss of generality since one can break ties randomly in an
online fashion.

Refer to the first
$\lceil m / \ell \rceil$
boys as the \emph{filter} boys and leave them unmatched.
Let $b_{f}$ be the heaviest filter boy and define the random set
$X = \{ b \in B \mid w(b) > w(b_{f}) \}$.
Observation~\ref{observation:good-event} ensures that the event
$n = \ell / 5 \leq |X| < \ell = 5 n$
occurs with probability at least
$1 / 13$
(recall the assumption that the boys' weights are distinct);
condition hereafter on this event.
This means, in particular, that
$w(X) \geq w(H_{B})$.

Let
$Y \subseteq B$
be the subset consisting of the first $n$ boys to arrive from $X$.
We run the algorithm promised by the assumption on all girls in $G$ and the
boys in $Y$, ignoring all remaining boys (leaving them unmatched).
The assertion follows since the random arrival order ensures that in
expectation,
$w(Y) \geq w(X) / 5$;
employing Markov's inequality, we can condition on $w(Y)$ being sufficiently
close to it.
\end{proof}

\begin{lemma} \label{lemma:girls-equivalent-to-boys}
There exists a conservative algorithm that approximates optimization criterion
\CritG{} within ratio
$\rho = \rho(n)$
in an instance with
$|G| = |B| = n$
if and only if
there exists a conservative algorithm that approximates optimization criterion
\CritB{} within ratio
$\rho = \rho(n)$
in an instance with
$|G| = |B| = n$.
This holds for every arrival order distribution.
\end{lemma}
\begin{proof}
Consider some instance that consists of girl set $G$ and boy set $B$.
We construct its
\emph{transposed} instance by setting
the girl set
$\bar{G} = \{ \bar{g} \mid g \in G \}$,
boy set
$\bar{B} = \{ \bar{b} \mid b \in B \}$,
and define the preferences over $\bar{G}$ and $\bar{B}$ so that
$\bar{g}_{1} \succ \bar{g}_{2}$
if and only if
$g_{1} \prec g_{2}$
and
$\bar{b}_{1} \succ \bar{b}_{2}$
if and only if
$b_{1} \prec b_{2}$.
Given some perfect matching $M$ between $G$ and $B$, construct its
\emph{transposed} matching by setting
$\bar{M} = \{ (\bar{g}, \bar{b}) \mid (g, b) \in M \}$.
For an individual
$x \in G \cup B$,
let $M(x)$ denote the individual to which $x$ is matched under $M$;
likewise, for an individual
$\bar{x} \in \bar{G} \cup \bar{B}$,
let $\bar{M}(\bar{x})$ denote the individual to which $\bar{x}$ is matched
under $\bar{M}$.

We argue that boy
$b \in B$
is satisfied under $M$ if and only if girl
$\bar{M}(\bar{b}) \in \bar{G}$
is satisfied under $\bar{M}$.
The assertion follows since the transposed instance can be constructed in an
online fashion and since the transposed instance of the transposed instance is
the original instance.
Indeed,
\begin{align*}
& ~
\text{$b$ is satisfied under $M$} \\
\Longleftrightarrow & ~
M(g) \succ b \quad \forall g \succ M(b) \\
\Longleftrightarrow & ~
\bar{M}(\bar{g}) \prec \bar{b} \quad \forall \bar{g} \prec \bar{M}(\bar{b}) \\
\Longleftrightarrow & ~
\text{$\bar{M}(\bar{b})$ is satisfied under $\bar{M}$} \, ,
\end{align*}
where the first and third transitions follow directly from the definition of a
satisfied individual and the second transition follows from the construction
of the transposed instance and matching.
\end{proof}

\section{Random Arrival Order}

\subsection{Maximizing the Number of Satisfied Individuals}
\label{section:individuals}

\begin{theorem} \label{theorem:positive-Cg-Cb}
Optimization criteria \CritG{} and \CritB{} (maximizing the number of
satisfied girls and boys) can be approximated within a (positive) constant
ratio.
\end{theorem}

Theorem~\ref{theorem:positive-Cg-Cb} is established by combining the following
lemma with Lemmas \ref{lemma:reduce-to-balanced-girls},
\ref{lemma:reduce-to-balanced-boys}, and
\ref{lemma:girls-equivalent-to-boys}.

\begin{lemma}\label{lemma:boys-lower-bound}
For every
$\epsilon>0$,
there exists
$n_0\in\mathbb N$
such that for any
$n\geq n_0$,
\DM{} has a conservative strategy that with probability at least
$1-\epsilon$,
satisfies at least
$(1/5-\epsilon)n$
boys in any instance with
$|G| = |B| = n$.
\end{lemma}
\begin{proof}
Fix
$0<\gamma<1/5$
and
$n\in\mathbb N$.
We describe a probabilistic algorithm for \DM{}.
We then claim that the probability that there are at least
$\gamma n$
satisfied boys converges to $1$ as $n$ goes to infinity.
Obviously, there is a deterministic algorithm (in the support of our
algorithm) that ensures at least the same guarantee.
Figure~\ref{figure:boys-algorithm} illustrates a typical output of the
algorithm.

Set $\delta=1/5-\gamma$ and $a=
2\gamma + \delta$.
Let $X=(R,W,Y)$ be a
multinomial random variable with
parameters $(a,a,1-2a;n)$.
Namely, $X$ can be
realized as follows: take $x_1,\ldots,x_n$ i.i.d.\ random variables taking
values in $\{\text{red},\text{white},\text{yellow}\}$ with probabilities
$\Pr(x_1=\text{red})=\Pr(x_1=\text{white})=a$ and
$\Pr(x_1=\text{yellow})=1-2a$; let $R=|\{i:x_i=\text{red}\}|$,
$W=|\{i:x_i=\text{white}\}|$, and
$Y=|\{i:x_i=\text{yellow}\}|$.
Each
realization of $X$ prescribes a deterministic algorithm parameterized by
$(R,W,Y)$.

We call the first $R$ boys to arrive ``red'' boys and index them in
decreasing order of preference $b_1\succ b_2 \succ \cdots \succ
b_{R}$.
The girls are indexed in decreasing order of
preference,
$g_1\succ g_2\cdots \succ g_n$.
The red boys are matched with the least preferred girls $g_n,\ldots,g_{n-R+1}$ in an arbitrary order (say, in order of their arrival).

Each one of the remaining boys $x$ is associated a number $rank(x)\in\{0,\ldots,R\}$ according to how he compares with the red boys,
\[
rank(x)=
\begin{cases}
0&\text{if $x\succ b_1$,}\\
i&\text{if $b_i\succ x\succ b_{i+1}$,}\\
R&\text{if $b_R \succ x$.}
\end{cases}
\]

We call the boys arriving from time $R+1$ until $R+W$ ``white.'' Let $r=\lceil
\tfrac 1 4 \delta n\rceil$.
We try to match as many white boys as
possible with the $R-r$ most
preferred girls while preserving the preference order.
In order to be able to do so we need
to assume that the $R-r$ most preferred girls are unmatched yet.
Therefore, if $2R-r>n$ the algorithm reports Catastrophe of Type I and halts.

When a white boy $x$ arrives we
match him either with
$g_{rank(x)-r}$ if $rank(x)> r$ and
$g_{rank(x)-r}$ is unmatched yet, or
with the least preferred unmatched girl $g_i$.
In the latter case, if $i \leq R-r$ the algorithm reports Catastrophe of Type
II and halts.

We call the boys arriving after time
$R+W$ ``yellow'' boys.
When a yellow boy $x$ arrives we
match him with the most preferred
unmatched girl $g_i$ subject to $i\geq rank(x)-r$.
I.e.,
$i=\min\{j\in [n]\setminus
[rank(x)-r-1]: \text{$g_j$ is
unmatched yet}\}$.
If that set is empty, the algorithm
reports Catastrophe of Type III and halts.

We now turn to analyze the number of
satisfied boys.
The idea is to estimate the number of white boys $x$ who are matched according to their rank to $g_{rank(x)-r}$, and to show that these boys are all satisfied.
Let $R'=\min\{R,\left\lceil(a-\tfrac{1}{4}\delta)n\right\rceil\}$ and
\[
R''=\left\vert[R']\setminus\{rank(x):\text{$x$ is a white boy}\}\right\vert.
\]
Intuitively, $R''$ approximates the
size of the complement of the image
of the mapping $x\mapsto rank(x)$, where $x$ ranges over the white boys.
Clearly, $R''$ bounds from above the number of unmatched girls among the $R-r$ most preferred girls at time $R+W$.

Consider the following (bad) events:
\begin{align*}
E_1&=\{\text{Catastrophe of Type I reported}\},\\
E_2&=\{\text{Catastrophe of Type II reported}\},\\
E_3&=\{\text{Catastrophe of Type III reported}\},\\
E_4&=\{R<(a-\tfrac{1}{4}\delta)n\},\\
E_5&=\{R''> \tfrac 1 2 an\}.
\end{align*}

We will soon show that for all $i=1,...,5$, $\Pr(E_i)\to 0$, as $n\to\infty$.
We now show that given that none of
the five bad events occurred, the
number of satisfied boys is at least
$\gamma n$.
Given that $E_1$ and $E_2$ do not
occur, the boys matched with girls
in $\{g_1,\ldots,g_{R-r}\}$ until
time $R+W$ are exactly all the white
boys $x$ whose match is
$g_{rank(x)-r}$.
Given that $E_3$ does not occur, these boys end up being satisfied (when the algorithm terminates). Indeed, if a white boy of rank $i$ is matched with $g_{i-r}$ and $j<i-r$, then $g_j$ is matched with either a white or a yellow boy. In the former case, $g_j$ is matched with a (white) boy of rank $j+r<i$. In the latter case, $g_j$ is matched with a yellow boy whose rank is at most $j+r$.

We show that the number of these boys
\[
W'=\left\vert\left([R]\setminus [r]\right)\cap \{rank(x):\text{$x$ is a white boy}\}\right\vert
\]
is at least $\gamma n$. Indeed, given that none of events $E_4$ or $E_5$ occurred, since $W'\geq R'-r-R''$,
\[
\frac {1}{n}W'\geq (a-\tfrac{1}{4}\delta) - \tfrac 1 4 \delta - \tfrac 1 2 a=\gamma.
\]

It remains to verify that for all $i=1,...,5$, $\Pr(E_i)\to 0$, as $n\to\infty$. By the weak law of large numbers, $\frac 1 n R\to a$ in probability, readily implying that $\Pr(E_1)$ and $\Pr(E_4)$ vanish as $n$ grows.

The following observation will be useful: let $x_1$ be the color of the most preferred boy, $x_2$ the color of the second most preferred boy, and so on until $x_n$. The random variables $x_1,\ldots,x_n$ are i.i.d.\ with $\Pr(x_i=\text{red})=\Pr(x_i=\text{white})=a$, and $\Pr(x_i=\text{yellow})=1-2a$. For the sake of argumentation we extend $x_1,x_2,\ldots$ to be an infinite sequence of i.i.d.\ random variables.

Let $t_1<t_2<\cdots$ be the occurrences of ``red,'' namely, $\{t_i\}_{i\in\mathbb N}=\{j:x_{j}=\text{red}\}$. Let $I_i$ be the indicator of the event $\{x_j\neq \text{white}, \forall\, t_i<j<t_{i+1}\}$. Clearly,
\[
R''=\sum_{i=1}^{R'}I_i\leq \sum_{i\leq (a-\tfrac{1}{4}\delta)n}I_i.
\]
Since $I_1,I_2,\ldots$ are i.i.d.\ $Bernoulli(\tfrac 1 2,\tfrac 1 2)$, by the weak law of large numbers, \[
\lim_{n\to\infty}\Pr(\sum_{i\leq (a-\tfrac{1}{4}\delta)n}I_i>\tfrac 1 2 an)=0.
\]
Therefore, $\Pr(E_5)$ vanishes as $n$ grows.

We show that $\Pr(E_2)$ vanishes by showing that $\Pr(E_2\setminus E_5)$ vanishes. Suppose we modified the algorithm so that when a Catastrophe of Type II occurs the algorithm skips the current boy (leaving him unmatched) and continues to the next boy. Consider the situation at time $R+W$ in the event $E_2\setminus E_5$. The girls $\{g_i\}_{j=R-r+1}^n$ are all matched. Among the other girls there are at most $R''+|R-(a- \tfrac 1 4\delta)n|$ unmatched girls. Therefore,
\begin{align*}
Y=1-R-W&<\text{``\#unmatched boys''}\\
&=\text{``\#unmatched girls''}\\
&\leq \tfrac 1 2 an +|R-(a- \tfrac 1 4\delta)n|,
\end{align*}
where the last inequality holds since we are in the case that event $E_5$ has not occurred.
By the weak law of large numbers, $\tfrac 1 n R\to a$ and $\tfrac 1 n Y\to 1-2a$ in probability. Since $1-2a=\tfrac 1 2 a + 2\tfrac 1 2 \delta> \tfrac 1 2 a + \tfrac 1 4 \delta$, it follows that $\Pr(E_2\setminus E_5)\to 0$, as $n\to \infty$.

It remains to show that $\Pr(E_3)$
vanishes.
To this end define
$Y_i=|\{t_i<j<t_{i+1}:x_j=\text{yellow}\}|$,
i.e., the number of yellow boys
between the $i$th and $(i+1)$th red boys.
The distribution of $Y_i+1$ is
geometric with success probability $a/(1-a)$, namely,
$\Pr(Y_i=k)=\left(\frac
{1-2a}{1-a}\right)^k \frac {a}{1-a}$
($k=0,1,\ldots$); hence
$\mathbb E[Y_i]=(1-a)/a -1 > \tfrac
1 2$.
Consider the i.i.d.\ random
variables $Z_i=Y_i-I_i$.
By the strong law of large numbers
$\tfrac 1 n \sum_{i=1}^n Z_i\to E[Z_1]>0$, almost surely. It follows that
\begin{equation*}
\lim_{n\to\infty}\Pr(\exists\, k\geq \tfrac 1 4 \delta n\text{ s.t. }\sum_{i=1}^k Z_i\leq 0) =0.
\end{equation*}

We show that $E_3 \subset \{\exists\, k\geq \tfrac 1 4 \delta n\text{ s.t. }\sum_{i=1}^k Z_i\leq 0\}$. Suppose that the algorithm reports Catastrophe of Type III upon the arrival of a yellow boy $x$ at time $t$. Let $k=\max\{i+r:\text{$g_i$ is unmatched at time $t$}\}$. Since it is a Catastrophe of Type III, we have $k < rank(x)$. The girls $g_i$, $i={k-r+1},\ldots,n$, are all matched with boys who are either red or white, or have rank at least $k+1$. To see this, consider any such $g_i$ who is matched with some yellow boy $y$. If $rank(y)<i+r$, then at the time $y$ arrived, all the girls $g_j$, $j={rank(y)-r,\ldots,i-1}$, were matched; therefore, at the present time, all the girls  $g_j$, $j= rank(y)-r,\ldots,n$, are matched; therefore $k< rank(y)$. Since the number of unmatched boys is equal to the number of unmatched girls at the time just before the catastrophe, we have
\begin{align*}
\sum_{i=1}^{k} I_i &\geq \sum_{i=r}^{k} I_i \\
&= \text{``\#girls who are either unmatched or matched to yellow boys of rank at most $k$''}\\
&= \text{``\#yellow boys who are either unmatched or have rank at most $k$''}\\
&> \text{``\#yellow boys who have rank at most $k$''}\\
&\geq \sum_{i=1}^{k} Y_i.
\end{align*}
It follows that $\sum_{i=1}^k Z_i <0$, and the proof is concluded since $k>r\geq \tfrac 1 4 \delta n$.
\end{proof}

\begin{figure}
\begin{center}
\includegraphics[width=0.5\textwidth]{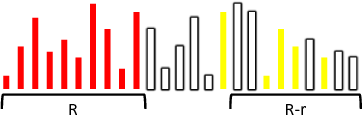}
\caption{\label{figure:boys-algorithm}
A typical matching in the proof of Lemma~\ref{lemma:boys-lower-bound}.
The vertical lines represent the boys. The first boy on the right is matched
with the most preferred girl $g_1$, the second boy with $g_2$, the third with
$g_3$, and so on. The lengths of the lines represents the quality of the
corresponding boys. The shorter the line the more preferred the boy is. A line
corresponds to a satisfied boy if there are no longer lines to its right. The
white boys on the $R-r$ right segment are all satisfied. There are roughly
$1/5n$ such boys.}
\end{center}
\end{figure}

\subsection{Maximizing the Number of Satisfied Matched Pairs}
\label{section:matched-pairs}

\begin{theorem} \label{theorem:negative-Cp}
Optimization criterion \CritP{} cannot be approximated within ratio better than
$O (1 / \sqrt{n})$
even in balanced instances with
$|G| = |B| = n$.
\end{theorem}
\begin{proof}
We establish the assertion for conservative algorithms;
the proof for general algorithms follows by
Lemma~\ref{lemma:conservative-Cb-Cp}.
Consider a two stage auxiliary game in which \DM{} is granted more power than
in the actual game.
We assume, for simplicity, that $n$ is even.
Let $R$ be a random set of $n/2$ boys.
Let us call the boys in $R$ ``red'' and the remaining boys ``white''.
In the first stage the red boys arrive (along with their preference order),
all at once, and \DM{} has to match them with $n/2$ girls.
In the second stage, the white boys arrive (along with their preference
order), at all once, and \DM{} has to match them with the remaining $n/2$
girls.
The objective is to maximize the expected number of satisfied pairs.

Denote the value of the auxiliary game $a(n)$. Since any strategy of the original game can be employed in the auxiliary game, the value of the original game is bounded from above by $a(n)$. We show that $a(n)\leq \sqrt{n\pi/2}+o(\sqrt n)$.

We restrict attention to a subset of the strategies of the auxiliary game. A \emph{simple} strategy in the auxiliary game is a strategy of the following form: (i) choose a set of $n/2$ girls $A$; (ii) match the red boys with $A$ in order of preference; (iii) match the white boys with the remaining girls in order of preference.

We show that any strategy of the auxiliary game is weakly dominated by a simple strategy. Take any pair of boys $b \succ b'$ and any pair of girls $g\succ g'$. Suppose there is a positive probability to the event that $b$ and $b'$ have the same color, and \DM{} matches $b$ with $g'$ and $b'$ with $g$. Modify \DM{}'s strategy, such that in the above event, \DM{} matches $b$ with $g$ and $b'$ with $g'$. By applying this modification the number of satisfied pairs cannot decrease. Indeed, any matched pair that involves a girl who is between $g$ and $g'$ in order of preference is already unsatisfied before the modification, and any other matched pair is unaffected by the modification. Repeatedly applying this modification, for any $b$, $b'$, $g$, and $g'$, yields a weakly dominating simple strategy.

Order the boys and the girls in order of preference $b_1\succ b_2\cdots\succ b_n$ and $g_1\succ g_2\cdots\succ g_n$. Since there is no advantage in randomizing, \DM{} has an optimal simple strategy in which she chooses a fixed $A\subset [n]$, matches the red boys with the $A$-indexed girls and the white boys with the remaining girls.

We estimate the number of satisfied pairs. Let $m\colon [n]\to [n]$, a matching from boys to girls, be the output of \DM{}'s strategy. With tolerable abuse of notation, we also use $m(S)$ to denote the set of girls matched to boys in $S$. By definition, a pair $(g_i,b_j)$ ($m(j)=i$) is satisfied if and only if
 \begin{align*}
m([j-1])&\subseteq [i-1], &&\text{(better boys mate better girls)}\\
m^{-1}([i-1])&\subseteq [j-1].&&\text{(better girls mate better boys)}
\end{align*}
Since $m$ is injective, the above implies that
\begin{align*}
m(i)&=i,\\
m([i-1])&=[i-1],\\
m([n]\setminus[i])&=[n]\setminus[i].\\
\end{align*}
Let $R\subset [n]$ be the indexes of the red boys. The set of girls who are matched with $R$ is predetermined $m(R)=A$. Therefore, in order for a girl $g_i$ to take part in a satisfied pair it is necessary (and sufficient) that the following event $E_i$ occurs
\begin{align*}
|R\cap [i-1]|&=|A\cap [i-1]|,\\
|R\setminus[i]|&=|A\setminus[i]|.
\end{align*}

By counting the values of $R$ that result in $E_i$,
\[
\Pr(E_i)=\frac{\binom{i-1}{|A\cap[i-1]|}\binom{n-i}{|A\setminus[i]|}}{\binom{n}{n/2}}\leq \frac{\binom{i-1}{\lceil (i-1)/2\rceil}\binom{n-i}{\lceil (n-i)/2\rceil}}{\binom{n}{n/2}}.
\]
Thus,
\[
a(n)\leq \sum_{i=1}^n \frac{\binom{i-1}{\lceil (i-1)/2\rceil}\binom{n-i}{\lceil (n-i)/2\rceil}}{\binom{n}{n/2}}.
\]
By Stirling's approximation,
\[
a(n) \lesssim \sum_{i=2}^{n-1} \frac{1}{\sqrt{2\pi}}\frac{\sqrt n}{\sqrt{(i-1)(n-i)}}\lesssim \sqrt n\int_0^1\frac{\mathrm{d}x}{\sqrt{2\pi x(1-x)}}=\sqrt{n\pi/2}.
\]
\end{proof}

On the positive side, using similar ideas to the ones applied in the original
secretary problem, one can guarantee an expected number of satisfied pairs of
$\frac 2 e -\epsilon$
(the proof of the following observation is deferred to
Appendix~\ref{appendix:missing-proofs}).

\begin{observation}\label{observation:positive-Cp}
For every
$\epsilon>0$,
there is
$n_0\in \mathbb N$
such that for every
$n\geq n_0$,
\DM{} can guarantee an expected number of satisfied pairs of
$\frac 2 e -\epsilon$
in any balanced instance with
$|G| = |B| = n$.
\end{observation}

\section{Adversarial Arrival Order}
\label{section:adversarial-arrival-order}
The instances considered in this section are balanced
($|G| = |B| = n$)
with an adversarial arrival order.
In every matching there is at least one satisfied girl (resp., boy);
indeed, the girl (resp., boy) that is matched to the most preferred boy
(resp., girl) is clearly satisfied.
The first observation of this section states that a deterministic \DM{} cannot
do any better than satisfying this single individual.

\begin{observation}
There is no deterministic \DM{} that satisfies more than a single girl in
balanced instances with an adversarial arrival order.
\end{observation}
\begin{proof}
Let $t^*$ denote the first time at which \DM{} matches boy $\pi(t^*)$ to the
least preferred girl or leaves a boy unmatched.
The adversary, knowing this in advance, provides a sequence of boys, so that
$\pi(t+1) \succ \pi(t)$
for every
$t < t^*$,
and
$\pi(t^*) \succ \pi(t)$
for every
$t > t^*$
(with an arbitrary order between them).
Since $\pi(t^*)$ is the strongest boy, all girls (with the exception of the
one matched with $\pi(t^*)$, if he is matched) form a blocking pair with him,
hence the result.
\end{proof}

\begin{corollary}
There is no deterministic \DM{} that satisfies more than a single boy in
balanced instances with an adversarial arrival order.
\end{corollary}
\begin{proof}
Follows directly from Lemmas \ref{lemma:conservative-Cb-Cp} and
\ref{lemma:girls-equivalent-to-boys}.
\end{proof}

We now turn our attention to a randomized \DM{}, proving that the situation is
still much worse than the one in instances with uniform random arrival order.

\begin{theorem} \label{theorem:negative-adversarial-Cg-Cb}
Under adversarial arrival order,
optimization criteria \CritG{} and \CritB{} cannot be approximated within
ratio better than
$O (1 / \sqrt{n})$
even in balanced instances with
$|G| = |B| = n$.
\end{theorem}

\begin{theorem} \label{theorem:negative-adversarial-Cp}
Under adversarial arrival order,
optimization criterion \CritP{} cannot be approximated within ratio better
than
$1 / n$
even in balanced instances with
$|G| = |B| = n$.
\end{theorem}

The proofs of Theorems \ref{theorem:negative-adversarial-Cg-Cb} and
\ref{theorem:negative-adversarial-Cp} are based on (the trivial direction of)
Yao's minimax principle.
The former is established by merging Lemma~\ref{lemma:adversarial-girls} with
Lemmas \ref{lemma:conservative-Cb-Cp} and \ref{lemma:girls-equivalent-to-boys}
and the latter follows from
Lemma~\ref{lemma:adversarial-pairs} (whose proof is deferred to
Appendix~\ref{appendix:missing-proofs}).

\begin{lemma} \label{lemma:adversarial-girls}
There exists a distribution $D$ over the instances, such that no
deterministic \DM{} can satisfy more than $\sqrt{2n}$ girls in
expectation when provided with a $D$-random instance.
\end{lemma}

\begin{lemma} \label{lemma:adversarial-pairs}
There exists a distribution $D$ over the instances, such that no
deterministic \DM{} can satisfy more than $1$ pair in
expectation when provided with a $D$-random instance.
\end{lemma}

The proofs of Lemmas \ref{lemma:adversarial-girls} and
\ref{lemma:adversarial-pairs} rely on a similar construction.
A sequence of probabilities
$p_2,\ldots,p_n\in[0,1]$ defines a distribution over permutations of the boys
$D=D(p_2,\ldots,p_n)$ as follows: the first boy $\pi(1)$ is either the most
preferred boy, with probability $p_n$, or the least preferred boy, with
probability $1-p_n$. Any subsequent boy $\pi(k)$ ($k<n$) is either the most
preferred boy among the remaining boys $\{\pi(k),\pi(k+1),\ldots,\pi(n)\}$,
with probability $p_{n+1-k}$, or the least preferred boy among the remaining
boys, with probability $1-p_{n+1-k}$. The rank of the last boy $\pi(n)$ is
already determined from the specification of the previous boys.

A key feature that makes $D$ hard to play against is that \DM{}'s information at time $k$, the relative order on $\{\pi(1),\ldots,\pi(k)\}$, is independent of the future, the relative order on $\{\pi(k),\ldots,\pi(n)\}$. That feature also simplifies the performance analysis, since we can refer to the expected number of satisfied boys/pairs among the last $n-k$ boys regardless of the assignment of the first $k$ boys.

\begin{proof}[Proof of Lemma~\ref{lemma:adversarial-girls}]
We use the distribution $D=D(p_2,\ldots,p_n)$, while specifying $p_1,\ldots,p_n$ recursively. Suppose $p_2,\ldots,p_k$ are already specified. Let $v_k$ be the expected number of satisfied boys under \DM{}'s best response to $D(p_2,\ldots,p_k)$. Clearly, $v_1=1$. Set
\[p_{k+1}=\frac{1}{1+v_k}.
\]
We show that
\begin{equation}\label{inequality:v_n}
v_{k+1}\leq v_k+\frac{1}{1+v_{k}},
\end{equation}
and deduce that $v_n < \sqrt{2n}$.

Assuming Inequality \eqref{inequality:v_n} we deduce that $1+v_k< \sqrt{2k}+\sqrt{2/k}$, $\forall k\in\mathbb N$. Define $u(x)=\sqrt{2x}+\sqrt{2/x}$. We show that
\begin{equation}\label{inequality:u}
u(x+1) > u(x)+\frac{1}{u(x)},\quad \forall x\geq 1.
\end{equation}
By Lagrange's mean value theorem $u(x+1)- u(x)=u'(\xi)$, for some
$\xi\in(x,x+1)$. Since $u'$ is decreasing, $u(x+1)- u(x)\geq u'(x+1)\geq
(2x+2)^{-\frac 1 2}$. Inequality~\eqref{inequality:u} follows, since
\[
\frac{1}{u'(x+1)}\leq \sqrt{2x+2} \leq \sqrt{2x}+\sqrt{2/x} = u(x).
\]
Where, the last inequality follows from Lagrange's mean value theorem and the fact that the derivative of the function $\sqrt{2x}$ is decreasing.

We must show that $1+v_k<u(k)$, $\forall k\in\mathbb N$. We do so by induction on $k$. The case $k=1$ holds since $v_1=1$. Assuming it hols for $k$,
\begin{align*}
1+v_{k+1}&\leq 1+v_{k}+\frac{1}{1+v_k} &&\text{(Inequality\eqref{inequality:v_n})}\\
        &< u(k)+\frac{1}{u(k)}       &&\text{(induction hypothesis; $x\mapsto x+\frac 1 x$ increases on $x\geq 1$)}\\
        &<u(k+1).                    &&\text{(Inequality\eqref{inequality:u})}
\end{align*}

It remains to prove Inequality \eqref{inequality:v_n}.
Consider two cases:
(i)
\DM{} matches $\pi(1)$ with the least preferred girl or possibly leaves him
unmatched;
(ii)
\DM{} matches $\pi(1)$ with another girl.
We show that in either cases the expected number of satisfied boys is at most
$v_k+\frac{1}{1+v_{k}}$.

Case (i):
Conditioned on the event that $\pi(1)$ is the least preferred boy, the
expected number of satisfied girls is at most $1+v_k$.
It is exactly $1+v_k$ when $\pi(1)$ is matched with the least preferred girl
and it is $v_k$ when he is left unmatched.
Conditioned on the event that $\pi(1)$ is the most preferred boy, no matter
whether he is matched or not, any girl he is not matched with is unsatisfied
and so there is at most one satisfied girl.
Therefore, the expected number of satisfied girls in case (i) is at most
\[
(1-p_{k+1})(1+v_k)+p_{k+1}=v_k+\frac{1}{1+v_{k}}.
\]

Case (ii):
Conditioned on the event that $\pi(1)$ is the least preferred boy, the girl
that $\pi(1)$ is matched with is not satisfied and there are at most $v_k$
satisfied girls in expectation.
Conditioned on the event that $\pi(1)$ is the most preferred boy, the girl he
is matched with is satisfied and in addition there are at most $v_k$ other
satisfied girls in expectation.
Therefore, the expected number of satisfied girls in case (ii) is at most
\[
(1-p_{k+1})v_k+p_{k+1}(1+v_k)= v_k+\frac{1}{1+v_{k}}.
\]
The assertion follows.
\end{proof}

\section{The Weighted case}

In this section we return to uniform random arrival orders and establish the
following theorems.

\begin{theorem} \label{theorem:positive-wCg}
Optimization criterion \CritGw{} can be approximated within ratio
$\Omega (1 / \log n)$.
\end{theorem}

\begin{theorem} \label{theorem:positive-wCb}
Optimization criterion \CritBw{} can be approximated within a (positive)
constant ratio.
\end{theorem}

\begin{proof}[Proof of Theorem~\ref{theorem:positive-wCg}]
Let $g^{*}$ be a heaviest girl in $H_{G}$ (which is also a heaviest girl
in $G$).
Partition $H_{G}$ into into \emph{weight classes}
$C_{1}, C_{2}, \dots$
so that
\[
C_{i}
\, = \,
\left\{
g \in H_{G} \mid w(g^{*}) / 2^{i} < w(g) \leq w(g^{*}) / 2^{i - 1}
\right\} \, .
\]
Taking
$k = O (\log n)$,
we observe that
$w \left( \bigcup_{i > k} C_{i} \right)
\leq
w(g^{*})$,
hence
$w \left( C_{1} \cup \cdots \cup C_{k} \right) \geq w(H_{G} / 2)$.

Let $i^{*}$ be an index
$1 \leq i \leq k$
that maximizes $w(C_{i})$.
Apply the algorithm promised by Theorem~\ref{theorem:positive-Cg-Cb}
(satisfying girls) to the problem instance that consists of the girls in
$C_{i^{*}}$ (whose weights are uniform up to factor $2$) and all boys in $B$;
the remaining girls are matched arbitrarily or left unmatched.
Theorem~\ref{theorem:positive-Cg-Cb} ensures that
$\Omega (|C_{i^{*}}|)$
girls in $C_{i^{*}}$ are satisfied as
$|C_{i^{*}}| \leq |H_{G}| \leq |B|$.
The assertion follows since
$w(C_{i^{*}}) \geq \Omega (w(H_{G}) / \log n)$.
\end{proof}

Theorem~\ref{theorem:positive-wCb} is established by combining the following
lemma with Lemma~\ref{lemma:reduce-to-balanced-boys}.

\begin{lemma}\label{lemma:boys-probability-lower-bound}
There exists a universal constant
$p > 0$
such that \DM{} has a strategy that satisfies each individual boy with
probability at least $p$ in any balanced instance
($|G| = |B|$).
\end{lemma}
\begin{proof}
The algorithm is similar to the algorithm in the proof of
Lemma~\ref{lemma:boys-lower-bound}.
The only difference is that here we set $r=0$ (instead of $\frac 1 4 \delta
n$).
As a result, we manage to guarantee a positive constant lower bound on the
probability of being satisfied \emph{for every single boy}.
Alas, the probability that at least $(\frac 1 5 - \delta)n$ boys are satisfied drop from being close to one to being merely bounded away from zero.

For completeness, we briefly repeat parts of the description of the algorithm
and other ideas that appear also in the proof of Lemma~\ref{lemma:boys-lower-bound}.
Fix $0<\gamma<1/5$ and $n\in\mathbb N$.
Set $\delta=1/5-\gamma$ and $a= 2\gamma + \delta$.
Let $X=(R,W,Y)$ be a multinomial random variable with parameters
$(a,a,1-2a;n)$.
Each realization of $X$ prescribes a deterministic algorithm parameterized by $(R,W,Y)$.

As before the first $R$ boys are called ``red,'' the next $W$ boys white, and the last $Y$ boys ``yellow.''If $2R>n$ we report Catastrophe of Type I.

The white boys are indexed in
decreasing order of preference $b_1\succ b_2 \succ \cdots \succ
b_{R}$.
The girls are indexed in decreasing order of
preference, $g_1\succ g_2\cdots \succ g_n$.
The red boys are matched with the least preferred girls $g_n,\ldots,g_{n-R+1}$ in an arbitrary order.

Each one of the remaining boys $x$ is associated a number $rank(x)\in\{0,\ldots,R\}$ according to how he compares with the red boys,
\[
rank(x)=
\begin{cases}
0&\text{if $x\succ b_1$,}\\
i&\text{if $b_i\succ x\succ b_{i+1}$,}\\
R&\text{if $b_R \succ x$.}
\end{cases}
\]

When a white boy $x$ arrives we match him either with $g_{rank(x)}$ if
$g_{rank(x)}$ is unmatched yet, or with the least preferred unmatched girl $g_i$.
In the latter case, if $i \leq R$ the algorithm reports Catastrophe of Type II and halts.

When a yellow boy $x$ arrives we match him with the most preferred unmatched
girl $g_i$ subject to $i\geq rank(x)$.
I.e., $i=\min\{j\in [n]\setminus [rank(x)-r-1]: \text{$g_j$ is unmatched
yet}\}$.
If that set is empty, the algorithm reports Catastrophe of Type III and halts.

Define
$R'=\min\{R,\left\lceil(a-\tfrac{1}{4}\delta)n\right\rceil\}$
and
\[
R''=\left\vert[R']\setminus\{rank(x):\text{$x$ is a white boy}\}\right\vert \,
,
\]
and consider the following (bad) events:
\begin{align*}
E_1&=\{\text{Catastrophe of Type I reported}\},\\
E_2&=\{\text{Catastrophe of Type II reported}\},\\
E_3&=\{\text{Catastrophe of Type III reported}\},\\
E_4&=\{R<(a-\tfrac{1}{4}\delta)n\},\\
E_5&=\{R''> \tfrac 1 2 an\}.
\end{align*}

The proof that $\Pr(E_1\cup E_2\cup E_4\cup E_5)\to 0$, as $n\to\infty$, follows the same lines as in the proof of Lemma~\ref{lemma:boys-lower-bound}.

Unlike the proof of Lemma~\ref{lemma:boys-lower-bound}, here $\Pr(E_3)$ is merely bounded away from one, rather than close to zero.

Let $x_1,x_2,\ldots$ be the colors of the boys in decreasing order of
preference extended to an infinite sequence of i.i.d.\ random variables.
Let $t_1<t_2<\cdots$ be the occurrences of ``red,'' namely,
$\{t_i\}_{i\in\mathbb N}=\{j:x_{j}=\text{red}\}$.
Let $I_i$ be the indicator of the event $\{x_j\neq \text{white}, \forall\,
t_i<j<t_{i+1}\}$.
Define $Y_i=|\{t_i<j<t_{i+1}:x_j=\text{yellow}\}|$, i.e., the number of yellow
boys between the $i$th and $(i+1)$th red boys.
The distribution of $Y_i+1$ is geometric with success probability $a/(1-a)$,
namely, $\Pr(Y_i=k)=\left(\frac {1-2a}{1-a}\right)^k \frac {a}{1-a}$ ($k=0,1,\ldots$); hence $\mathbb E[Y_i]=(1-a)/a -1 > \tfrac 1 2$.
Consider the i.i.d.\ random variables $Z_i=Y_i-I_i$.

We show that
\begin{equation}\label{equation:biased-random-walk}
\Pr(\forall k\in\mathbb N\ \sum_{i=1}^k Z_i> 0)>0.
\end{equation}
By the strong law of large numbers $\tfrac 1 n \sum_{i=1}^n Z_i\to E[Z_1]>0$,
almost surely.
It follows that there is $l\in\mathbb N$ such that
$\Pr(\forall k\ \sum_{i=1}^k Z_i> -l)>0$.
Since
$\Pr(Z_1\geq 1)>0$,
\begin{multline*}
\Pr(\forall k\in\mathbb N\ \sum_{i=1}^k Z_i> 0)\\
\geq \Pr(Z_1,\ldots,Z_l\geq 1,\ \forall k\ \sum_{i=l+1}^{l+k} Z_i> -l)=\Pr(Z_1\geq 1)^l\Pr(\forall k\ \sum_{i=1}^{k} Z_i> -l)>0.
\end{multline*}

Next, the proof is concluded by showing that $E_3\subset\{\exists k\in\mathbb N\ \sum_{i=1}^k Z_i\leq 0\}$, since the probability of the latter event is smaller than one, by \eqref{equation:biased-random-walk}.

Suppose that the algorithm reports Catastrophe of Type III upon the arrival of
a yellow boy $x$ at time $t$.
Let $k=\max\{i:\text{$g_i$ is unmatched at time $t$}\}$.
Since it is a Catastrophe of Type III, we have $k < rank(x)$.
The girls $g_i$, $i={k-r+1},\ldots,n$, are all matched with boys who are
either red or white, or have rank at least $k+1$.
To see this, consider any such $g_i$ who is matched with some yellow boy $y$.
If $rank(y)<i$, then at the time $y$ arrived, all the girls $g_j$,
$j={rank(y),\ldots,i-1}$, were matched; therefore, at the present time, all the girls  $g_j$, $j= rank(y),\ldots,n$, are matched; therefore $k< rank(y)$.
Since the number of unmatched boys is equal to the number of unmatched girls at the time just before the catastrophe, we have
\begin{align*}
\sum_{i=1}^{k} I_i &  \\
&= \text{``\#girls who are either unmatched or matched to yellow boys of rank at most $k$''}\\
&= \text{``\#yellow boys who are either unmatched or have rank at most $k$''}\\
&> \text{``\#yellow boys who have rank at most $k$''}\\
&\geq \sum_{i=1}^{k} Y_i.
\end{align*}
It follows that $\sum_{i=1}^k Z_i <0$.
\end{proof}

\bibliographystyle{plain}
\bibliography{stabsec}

\clearpage
\appendix
\begin{center}
\Large{APPENDIX}
\end{center}

\section{Missing Proofs}
\label{appendix:missing-proofs}

\begin{proof}[Proof of Observation~\ref{observation:good-event}]
Fix
$q = \lceil k / \ell \rceil$
and observe that for every
$r \leq \ell$,
we have
\[
\Pr(R > r)
\, = \,
\frac{k - q}{k}
\cdot
\frac{k - q - 1}{k - 1}
\cdots
\frac{k - q - (r - 1)}{k - (r - 1)} \, .
\]
It follows that
\[
\Pr(R > r)
\, \leq \,
\frac{k - k / \ell}{k}
\cdot
\frac{k - k / \ell - 1}{k - 1}
\cdots
\frac{k - k / \ell - (r - 1)}{k - (r - 1)}
\, \leq \,
\left( 1 - \frac{1}{\ell} \right)^{r}
\, < \,
e^{-r / \ell}
\]
and
\[
\Pr(R > r)
\, \geq \,
\frac{k - k / \ell - 1}{k}
\cdot
\frac{k - k / \ell - 2}{k - 1}
\cdots
\frac{k - k / \ell - r}{k - (r - 1)}
\, > \,
\left( 1 - \frac{k / \ell + 1}{k - r} \right) \, .
\]
Taking $c$ to be sufficiently large so that
$r \leq \ell \leq k / 3$,
we ensure that
\[
\ell \leq k - 2 r
\, \Longleftrightarrow \,
k + \ell \leq 2 k - 2 r
\, \Longleftrightarrow \,
\frac{k / \ell + 1}{k - r} \leq \frac{2}{\ell} \, ,
\]
thus
\[
\Pr(R > r)
\, > \,
\left( 1 - \frac{2}{\ell} \right)^{r}
\, > \,
e^{-4 r / \ell} \, ,
\]
where the second transition follows by taking $c$ to be sufficiently large so
that
$2 / \ell \leq 0.79$.
Therefore,
\[
\Pr(\ell / 5 < R \leq \ell)
\, = \,
\Pr(R > \ell / 5) - \Pr(R > \ell)
\, > \,
e^{-4 / 5} - e^{-1}
\]
which establishes the assertion as
$e^{-4 / 5} - e^{-1} \approx 1 / 12.28$.
\end{proof}

\begin{proof}[Proof of Observation~\ref{observation:positive-Cp}]
Recall the classical secretary problem in which \DM{} has to stop upon the arrival of some $x$ and the objective is to maximize the probability that $x$ is the most preferred boy. The optimal strategy in the secretary problem is to wait until time $k\approx \tfrac 1 e n$, and then stop upon the first arrival of a boy who is more preferred than all of the previous boys. The probability of success converges to $\tfrac 1 e$, as $n$ grows.

From the solution to the secretary problem we device a matching strategy as follows: in the first $k=\lfloor\tfrac 1 e n\rfloor$ steps, match the boys with arbitrary girls who are neither the most preferred nor the least preferred girl. Continue in the same manner while reserving the most and least preferred girls for the first arrivals of boys who are either more preferred or less preferred than all previous boys. Upon the first arrival of a boy $x$ who is more preferred than all previous boys, match $x$ with the most preferred girl. Similarly, match the first boy who is less preferred than all previous boys with the least preferred girl. At times $n-1$ and $n$ match the arriving boys arbitrarily.

In any matching in which the most (resp. least) preferred boy and girl are matched together, they form a satisfied pair; therefore, by the guarantee of the secretary problem solution and the additivity of expectation, the proposed algorithm guarantees an expected number of $\frac 2 e -o(1)$ satisfied pairs.
\end{proof}

\begin{proof}[Proof of Lemma~\ref{lemma:adversarial-pairs}]
We establish the assertion for conservative algorithms;
the proof for general algorithms follows by
Lemma~\ref{lemma:conservative-Cb-Cp}.
We use the distribution $D(\frac 1 2,\ldots,\frac 1 2)$, i.e., each boy is
either more or less preferred than all of the boys that come after him with
equal probabilities.

Let $v_n$ denote the expected number of satisfied under an optimal online
assignment. We show that $v_{n+1}\leq\max\{v_n,\frac 1 2 (v_n +1)\}$. Since
$v_1=1$, we have, by induction on $n$, that $v_n\leq 1$, for all $n$.

We divide into two cases: (i) \DM{} matches $\pi(1)$ with either the most or least preferred girl; (ii) \DM{} matches $\pi(1)$ with with some other girl. Denote the rank of $\pi(1)$ by $r\in\{1,n+1\}$ (assuming there are $n+1$ boys and $n+1$ girls). We show that the expected number of satisfied pairs is at most $\frac 1 2 (v_n+1)$, in case (i), and $v_n$ in case (ii).

In Case (i), with probability $\frac 1 2$, $\pi(1)$ is matched with the girl of rank $r$. In this event they form a satisfied pair and the expected number of additional satisfied pairs is at most $v_n$. In the complement event that $\pi(1)$ is matched with the girl of rank $n+2-r$, none of the pairs is satisfied. Therefore, the expected number of satisfied pairs in case (i) is at most $\frac 1 2 (v_n+1)$.

In case (ii), $\pi(1)$ does not belong to a satisfied pair. The expected number of satisfied pairs among the remaining boys and girls is at most $v_n$. Therefore, the expected number of satisfied pairs in case (ii) is at most $v_n$.
\end{proof}

\end{document}